\documentclass{llncs}
\usepackage[nocenter]{qtree}
\usepackage{url,amssymb,graphicx,alltt,array,longtable,verbatim}
\usepackage{amssymb,amsfonts,proof,amsmath,mathtools,algpseudocode,algorithm,color,listings,enumitem}
\usepackage{verbatim}

\newcommand{\OMIT}[1]{}
\newcommand{\R}{\mathbb{R}}

\newcommand{\RealTh}{\ensuremath{\struct{\R,+,0,1}}}

\newcommand{\hi}[1]{\ensuremath{|#1|}}
\newcommand{\size}[1]{\ensuremath{\|#1\|}}
\newcommand{\shapet}[1]{\ensuremath{\langle #1 \rangle}}
\newcommand{\listc}[2]{\ensuremath{\textsf{concat}(#1,#2)}}
\newcommand{\hide}[1]{}
\newcommand{\fname}[1]{\ensuremath{\textsc{#1}}}
\newcommand{\lbow}[1]{\ensuremath{_{#1}\!{\bowtie}}}
\newcommand{\rbow}[1]{\ensuremath{\bowtie_{#1}}}
\def \defi {\ensuremath{\stackrel{\text{def}}{=}}}
\def \sleaf {\ensuremath{\ast}}
\def \bmodel {\ensuremath{\mathcal{A}}}

\def \lmodel {\ensuremath{\mathcal{C}}}
\def \tmodel {\ensuremath{\mathcal{K}}}

\def \treedom {\ensuremath{\mathbb{T}}}
\newcommand \tree[2] {\ensuremath{\Tree [ $#1$ $#2$ ]}}

\newcommand\fth[1]{\ensuremath{\mathsf{Th}(#1)}}

\def \elist {\ensuremath{\mathsf{nil}}}
\def \blmodel {\ensuremath{\mathcal{L}}}
\def \brmodel {\ensuremath{\mathcal{B}}}
\def \ltree {\ensuremath{\Tree [ $\bullet$ $\circ$ ]}}
\def \rtree {\ensuremath{\Tree [ $\circ$ $\bullet$ ]}}
\def \embed {\ensuremath{g}}
\def \utree {\ensuremath{\mathcal{U}(\mathbb{T})}}

\def \shapedom {\ensuremath{\mathbb{S}}}
\newcommand \ndex[1] {\ensuremath{#1{\mathsf{EXP}}}}
\newcommand{\overbar}[1]{\mkern 2.3mu\overline{\mkern-2.3mu#1\mkern-2.3mu}\mkern 2.3mu}

\newcommand{\struct}[1]{\ensuremath{\langle #1 \rangle}}

\begin{document}

\title{Complexity Analysis of Tree Share Structure}

\author{Xuan-Bach Le$^{\dagger}$ \quad Aquinas Hobor$^{\ddagger,*}$ \quad Anthony W. Lin$^{\dagger}$}

\institute{$^{\dagger}$University of Oxford \quad $^{\ddagger}$Yale-NUS College \quad $^{*}$National University of Singapore}
\maketitle

\begin{abstract}
The tree share structure proposed by Dockins et al. is an
elegant model for tracking disjoint ownership in
concurrent separation logic, but decision procedures for tree shares
are hard to implement due to a lack of a systematic theoretical study.
We show that the first-order theory of the full Boolean algebra
of tree shares (that is, with all tree-share constants) is decidable and has
the same complexity as of the first-order theory of Countable Atomless
Boolean Algebras.  We prove that combining this additive structure with
a constant-restricted unary multiplicative ``relativization'' operator has a
non-elementary lower bound.  We examine the consequences of this lower bound and prove
that it comes from the combination of both theories by proving an upper bound
on a generalization of the restricted multiplicative theory in isolation.
%The combined theory is thus an example of a non-elementary theory whose key
%subtheories are elementary.
\end{abstract}

\section{Introduction}
\label{sec:intro}

One general challenge in concurrent program verification is how to specify the ownership of shared resources among threads. A common solution is to tag shared resources with \emph{fractional shares}
that track ``how much'' of a resource is owned by an actor.  A \emph{policy} maps ownership quanta
with permitted behaviour.  For example, a memory cell can be ``fully owned'' by a thread, permitting 
both reading and writing; ``partially owned'', permitting only reading; or ``unowned'', permitting 
nothing; the initial model of fractional shares~\cite{Boyland:2003} was rationals in $[0,1]$.  Since their introduction, many program logics have used a variety of flavors of fractional 
permissions to verify programs~\cite{Boyland:2003,bornat05,parkinson:nbs,HoborG11,MSL,villard:phd,appel:programlogics,Le:2015,Caper:2016,Doko:2017,le17:logic,Dohrau18}.

Rationals do not mix cleanly with concurrent separation logic~\cite{OHearn:2007} because they do not preserve the ``disjointness'' property of separation logic~\cite{parkinson:phd}.  Dockins \emph{et al.}~\cite{dockins09:sa} proposed a ``tree share'' model that do preserve this property, and so a number of program logics have incorporated them~\cite{hobor08:phd,HoborG11,villard:phd,appel:programlogics,le17:logic}.

In addition to their good metatheoretic properties, tree shares have desirable computational
properties, which has enabled several highly-automated verification tools to incorporate them~\cite{villard:phd,hobor12:barriertool} via heuristics and decision procedures~\cite{le12:sharedec,le17:certproc}.  As we shall explain in \S\ref{sub:overview}, tree shares have both ``additive'' and ``multiplicative'' substructures.  All of the verification
tools used only a restricted fragment of the additive substructure (in particular, with only
one quantifier alternation) because the general theory's computational structure was not
well-understood.  These structures are worthy of further study both because even short programs 
can require hundreds of tree share entailment queries in the permitted formalism~\cite[Ch4:\S2,\S6.4,\S6.6]{cris:thesis}, and because recent program logics have
shown how the multiplicative structures aid program verification~\cite{appel:programlogics,le17:logic}.

Recently, Le \emph{et al.} did a more systematic analysis of the computational complexity of
certain classes of tree share formulae~\cite{le16:complex}; briefly:
\begin{itemize}
\item the additive structure forms a Countable Atomless Boolean Algebra, giving a well-understood complexity for all first-order formulae \emph{so long as they only use the distinguished constants} ``empty'' \textbf{0} and ``full'' \textbf{1};
\item the multiplicative structure has a decidable existential theory but an undecidable first-order theory; and
\item the additive theory in conjunction with a weakened version of the multiplicative theory---in particular, only permitting multiplication by constants on the
    right-hand side---regained first-order decidability.
\end{itemize}

\textbf{Contributions.} We address significant gaps in our theoretical understanding of tree 
shares that deter their use in automated tools for more sophisticated tasks.
\begin{itemize}
\item[\S\ref{sec:boolean}] Moving from a restricted fragment of a
first-order additive theory to the more general setting of unrestricted first-order formulae 
over Boolean operations is intuitively appealing due to the increased expressibility of the 
logic.  This expressibility even has computational consequences, as we demonstrate by using it
to remove a common source of quantifier alternations.  However, verifications in practice
often require formulae that incorporate more general constants than \textbf{0} and \textbf{1}, limiting the application of the analysis from~\cite{le16:complex} in practice.
This is unsurprising since it is true in other settings: many Presburger 
formulae that arise in engineering contexts, for example, are littered with application-specific
constants, \emph{e.g.}, $\forall x.(\exists y.x + y = 7) \Rightarrow (x + 13 < 21)$.  
A recent benchmark using tree shares for program verification~\cite{le17:certproc} supports this intuition: it made 16k calls in the supported first-order additive fragment, and 21.1\% (71k/335k) of the constants used in practice were neither \textbf{0} nor \textbf{1}.  Our main contribution
on the additive side is to give a polynomial-time algorithm that reduces first-order additive
formulae with arbitrary tree-share constants to first-order formulae using only \textbf{0} and \textbf{1}, demonstrating that the additive structure's exact complexity is $\mathsf{STA}(\ast,2^{n^{O(1)}},n)$-complete and closing the theory/practice gap between~\cite{le16:complex} and~\cite{le17:certproc}.
\item[\S\ref{sec:combresult}] We examine the combined additive/restricted multiplicative theory
proved decidable in~\cite{le16:complex}.  We prove a nonelementary lower bound
for this theory, via a reduction from the combined theory into the string structure with suffix 
successors and a prefix relation, closing the complexity gap in the theory. % of~\cite{le16:complex}.
\item[\S\ref{sec:bowresult}] We investigate the reasons
for, and mitigants to, the above nonelementary lower bound.  First, we show that the first-order
restricted-multiplicative theory on its own (\emph{i.e.}, without the Boolean operators) has
elementary complexity via an efficient isomorphism with strings equipped with prefix and suffix successors.  Thus, the nonelementary behavior comes precisely
from the combination of both theories.  Lastly, we examine the kinds of formulae that we expect
in practice---for example, those coming from biabduction problems discussed in~\cite{le17:logic}---and notice that they have elementary complexity.  
\end{itemize}
The other sections of our paper support our contributions by (\S\ref{sec:prelim}) overviewing tree shares, related work, and several basic complexity results; and by (\S\ref{sec:conclude}) discussing
directions for future work and concluding.

\section{Preliminaries}\label{sec:prelim}
Here we document the preliminaries for our result.  Some are standard (\S\ref{sub:comp}) while others are specific to the domain of tree shares (\S\ref{sub:overview}--\S\ref{sub:treepre}).

\subsection{Complexity preliminaries}
\label{sub:comp}

We assume that the readers are familiar with basic concepts in computational
complexity such as Turing machine, many-one reduction, space and time
complexity classes such as $\mathsf{NP}$ and $\mathsf{PSPACE}$. A problem is
\emph{nonelementary} if it cannot be solved by any deterministic Turing machine that can be time-bounded by one of the exponent functions $\mathsf{exp(1)} = 2^n, \mathsf{exp(n+1) = 2^{\mathsf{exp}(n)}}$. Let $\mathsf{A}$, $\mathsf{R}$ be complexity classes,
a problem $P$ is $\leq_{\mathsf{R}}$-complete for $\mathsf{A}$ iff $P$ is in
$\mathsf{A}$ and every problem in $\mathsf{A}$ is many-one reduced into $P$
via Turing machines in $\mathsf{R}$. In addition, we use
$\le_{\mathsf{R}\text{-lin}}$ to assert \emph{linear reduction} that belongs
to $\mathsf{R}$ and only uses linear space with respect to the problem's
size. In particular, $\le_{\text{log-lin}}$ is linear log-space reduction.
Furthermore, we denote $\mathsf{STA}(p(n),t(n),a(n))$ the class of
alternating Turing machine~\cite{Chandra:1981} that uses at most $p(n)$
space, $t(n)$ time and $a(n)$ alternations between universal states and
existential states or vice versa for input of length $n$. If any of the three
bounds is not specified, we replace it with the symbol $\ast$, \emph{e.g.}
$\mathsf{STA}(\ast,2^{n^{O(1)}},n)$ is the class of alternating Turing machines
that have exponential time complexity and use at most $n$ alternations.

\subsection{Overview of tree share structure}\label{sub:overview}

A tree share is a binary tree with Boolean leaves $\circ$ (white leaf) and
$\bullet$ (black leaf).  Full ownership is represented by $\bullet$ and no
ownership by $\circ$.  For fractional ownership, one can use, \emph{e.g.}
$\Tree [ $\bullet$ $\circ$ ] $, to represent the left half-owned resource.
Importantly and usefully, $\Tree [ $\circ$ $\bullet$ ] $ is a distinct tree
share representing the other right half. We require tree shares are in
canonical form, that is, any subtree $\Tree [ $\tau$ $\tau$ ]$ where $\tau
\in \{\bullet,\circ\}$ needs to be rewritten into $\tau$. For example, both
$\Tree [ $\bullet$ $\circ$ ]$ and $\Tree [ [ $\bullet$ $\bullet$ ] [ $\circ$
[ $\circ$ $\circ$ ] ] ]$ represent the same tree share but only the former
tree is canonical and thus valid. As a result, the set of tree shares
$\mathbb{T}$ is a strict subset of the set of all Boolean binary trees. Tree
shares are equipped with Boolean operators $\sqcup$ (union), $\sqcap$
(intersection) and $\bar{\cdot}$ (complement). When applied to tree shares of
height zero, \emph{i.e.} $\{\bullet,\circ\}$,  these operators  give the same
results as in the case of binary BA. Otherwise, our tree shares need to be
unfolded and folded accordingly before and after applying the operators
leaf-wise, \emph{e.g.}

$$
\overline{\Tree [ $\bullet$ $\circ$ ] } = \Tree [ $\circ$ $\bullet$ ] \quad\quad \Tree [ [ $\bullet$ $\circ$ ] $\bullet$ ] \sqcup \Tree [ $\circ$ [ $\circ$ $\bullet$ ] ] \cong  \Tree [ [ $\bullet$ $\circ$ ] [ $\bullet$ $\bullet$ ] ] \sqcup \Tree [ [ $\circ$ $\circ$ ] [ $\circ$ $\bullet$ ] ] = \Tree [ [ $\bullet$ $\circ$ ] [ $\bullet$ $\bullet$ ] ] \cong \Tree [ [ $\bullet$ $\circ$ ] $\bullet$ ]~.
$$
The additive operator $\oplus$ can be defined using $\sqcup$ and $\sqcap$, \emph{i.e.} disjoint union:

$$
a \oplus b = c ~\defi~ a \sqcup b = c \wedge a \sqcap b = \circ .
$$
Tree shares also have a multiplicative operator $\bowtie$ called ``bowtie'', where $\tau_1 \bowtie \tau_2$ is defined by replacing each black leaf $\bullet$ of $\tau_1$ with an instance of $\tau_2$, \emph{e.g.}

$$
\Tree [ [ $\bullet$ $\circ$ ] [ $\circ$ $\bullet$ ] ] ~ \bowtie  ~ \tree{\circ}{\bullet} ~~ = ~~ \Tree [ [ [ $\circ$ $\bullet$ ] $\circ$ ] [ $\circ$ [ $\circ$ $\bullet$ ] ] ]~.
$$

While the $\oplus$ operator has standard additive properties such as
commutativity, associativity and cancellativity, the $\bowtie$ operator
enjoys the unit $\bullet$, is associative, injective over non-$\circ$
arguments, and distributes over $\{\sqcup,\sqcap,\oplus\}$ on the
left~\cite{dockins09:sa}.
However, $\bowtie$ is not commutative, \emph{e.g.}:
\[
\Tree [ $\bullet$ $\circ$ ] \bowtie \Tree [ $\circ$ $\bullet$ ] ~ = ~
\Tree [ [ $\circ$ $\bullet$ ] $\circ$ ] ~ \neq ~
\Tree [ $\circ$ [ $\bullet$ $\circ$ ] ] ~ = ~
\Tree [ $\circ$ $\bullet$ ] \bowtie \Tree [ $\bullet$ $\circ$ ]
\]
The formalism of these binary operators can all be found
in~\cite{dockins09:sa}.

\subsection{Tree shares in program verification}\label{sub:program}

Fractional permissions in general, or tree shares in particular, are integrated into separation logic to reason about ownership. In detail, the mapsto predicate $x \mapsto v$ is enhanced with the permission $\pi$, denoted as $x \xmapsto{\pi} v$, to assert that $\pi$ is assigned to the address $x$ associated with the value $v$. This notation of fractional mapsto predicate allows us to split and combine permissions conveniently using the additive operator $\oplus$ and disjoint conjunction $\star$:

\begin{equation}
x \xmapsto{\pi_1 \oplus \pi_2} v \dashv \vdash x \xmapsto{\pi_1} v \star x \xmapsto{\pi_2} v.
\label{eq:fmap}
\end{equation}

The key difference between tree share model $\struct{\mathbb{T},\oplus}$ and rational model $\struct{\mathbb{Q},+}$ is that the latter fails to preserve the disjointness property of separation logic. For instance, while the predicate $x \mapsto 1 \star x \mapsto 1$ is unsatisfiable, its rational version $x \xmapsto{0.5} 1 \star x \xmapsto{0.5} 1$, which is equivalent to $x \xmapsto{1} 1$ by~(\ref{eq:fmap}), is satisfiable. On the other hand, the tree share version $x \xmapsto{\Tree [ $\bullet$ $\circ$ ]} \star x \xmapsto{\Tree [ $\bullet$ $\circ$ ]}$ remains unsatisfiable as the sum $\Tree [ $\bullet$ $\circ$ ] \oplus \Tree [ $\bullet$ $\circ$ ]$ is undefined. Such defect of the rational model gives rise  to the deformation of recursive structures or elevates the difficulties of modular reasoning, as first pointed out by~\cite{parkinson:phd}.

Recently, Le and Hobor~\cite{le17:logic} proposed a proof system for disjoint permissions using the structure $\struct{\treedom,\oplus,\bowtie}$. Their system introduces the notion of predicate multiplication where $\pi \cdot P$ asserts that the permission $\pi$ is associated with the predicate $P$. To split the permission, one can apply the following bi-entailment:

$$
\pi \cdot P \dashv \vdash (\pi \bowtie \Tree[ $\bullet$ $\circ$ ]) \cdot P \star (\pi \bowtie \Tree[ $\circ$ $\bullet$ ]) \cdot P.
$$

which requires the following property of tree shares to hold:
\begin{equation}
\label{eq:fragmentex}
\forall \pi. ~\pi = (\pi \bowtie \Tree[ $\bullet$ $\circ$ ]) \oplus(\pi \bowtie \Tree[ $\circ$ $\bullet$ ]).
\end{equation}

Note that the above property demands a combined reasoning of both $\oplus$
and $\bowtie$. While such property can be manually proved in theorem provers
such as Coq~\cite{Coq} using inductive argument, it cannot be handled
automatically by known tree share solvers~\cite{le12:sharedec,le17:certproc}
due to the shortness of theoretical insights.

\subsection{Previous results on the computational behavior of tree
shares}\label{sub:treepre}

The first sophisticated analysis of the computational properties of tree
shares were done by Le \emph{et al.}~\cite{le16:complex}. They showed that
the structure $\struct{\treedom, \sqcup, \sqcap, \bar{\cdot}}$ is a Countable
Atomless BA and thus is complete for the Berman complexity class
$\mathsf{STA}(\ast,2^{n^{O(1)}},n)$---problems solved by alternating
exponential-time Turing machines with unrestricted space and $n$
alternations---\emph{i.e.} the same complexity as the first-order theory over
the reals $\RealTh$ with addition but no multiplication~\cite{Berman}.
However, this result is restrictive in the sense that the formula class only
contains $\{\bullet,\circ\}$ as constants, whereas in practice it is
desirable to permit arbitrary tree constants, \emph{e.g.} $\exists a\exists
b.~a \sqcup b = \Tree [ $\bullet$ $\circ$ ]$.

When the multiplication operator $\bowtie$ is incorporated, the computational
nature of the language becomes harder.  The structure $\struct{\treedom, \bowtie}$---without
the Boolean operators---is isomorphic to word equations~\cite{le16:complex}.
Accordingly, its first-order theory is undecidable while its existential theory is decidable
with continuously improved complexity bounds currently at
$\mathsf{PSPACE}$ and $\mathsf{NP}$-hard (starting from Makanin's
argument \cite{Makanin} in 1977 and continuing with \emph{e.g.} \cite{Jez}).

Inspired by the notion of ``semiautomatic structures''~\cite{jain14:semi},
Le \emph{et al.}~\cite{le16:complex} restricted
$\bowtie$ to take only constants on the right-hand side, \emph{i.e.} to a
family of unary operators indexed by constants $\rbow{\tau}(x) \defi x \bowtie \tau$.
Le \emph{et al.} then examined
$ \lmodel \defi \struct{\treedom,\sqcup,\sqcap,\bar{\cdot},\rbow{\tau}}$.
Note that the verification-sourced sentence~\eqref{eq:fragmentex}
from \S\ref{sub:program} fits perfectly into $\lmodel$:
$\forall \pi.~ \pi ~ = ~ \rbow{ \Tree[ $\bullet$ $\circ$ ] } \! (\pi) ~ \oplus ~ \rbow{ \Tree[ $\circ$ $\bullet$ ] } \! (\pi)$.
Le \emph{et al.} encoded $\lmodel$ into \emph{tree-automatic structures}
\cite{blumensath04:as}, i.e., logical structures whose constants can be encoded
as trees, and domains and predicates finitely represented by tree automata. As 
a result, its first-order theory---with
arbitrary tree constants---is
decidable~\cite{blumensath04:as,blumensath99:as,anthony-thesis}, but until our 
results in
\S\ref{sec:combresult} the true complexity of $\lmodel$ was unknown.

\section{Complexity of Boolean structure $\bmodel \defi \struct{\mathbb{T},\sqcup,\sqcap,\bar{\cdot}}$}\label{sec:boolean}
\label{sec:baresult}
Existing tree share
solvers~\cite{le12:sharedec,le17:certproc} only utilize the additive operator
$\oplus$ in certain restrictive first-order segments. Given the fact that
$\oplus$ is defined from the Boolean structure $\bmodel =
\struct{\mathbb{T},\sqcup,\sqcap,\bar{\cdot}}$, it is compelling to establish
the decidability and complexity results over the general structure $\bmodel$.
More importantly, operators in $\bmodel$ can help reduce the complexity of a
given formula. For example, consider the following separation logic
entailment:

$$
a \xmapsto{\tau} 1 \star a \xmapsto{\Tree [ $\bullet$ [ $\bullet$ $\circ$ ] ]} ~\vdash~ a \xmapsto{\Tree [ $\bullet$ $\circ$ ] } 1 \star \top .
$$

To check the above assertion, entailment solvers have to extract and verify
the following corresponding tree share formula by grouping shares from same
heap addresses using $\oplus$ and then applying equality checks:

$$
\forall \tau \forall \tau'. \tau \oplus \Tree [ $\bullet$ [ $\bullet$ $\circ$ ] ] = \tau' \rightarrow \exists \tau''. \tau'' \oplus \Tree [ $\bullet$ $\circ$ ] = \tau' .
$$

By using Boolean operators, the above $\forall \exists$ formula can be
simplified into a $\forall$ formula by specifying that
either the share in the antecedent is not possible, or the share in the
consequent is a `sub-share' of the share in the antecedent:

$$
\forall \tau.~ \neg (\tau \sqcap \Tree [ $\bullet$ [ $\bullet$ $\circ$ ] ] = \circ) ~\vee~ ( \Tree [ $\bullet$ $\circ$ ] \sqsubseteq \tau \oplus \Tree [ $\bullet$ [ $\bullet$ $\circ$ ] ] ) .
$$

where the `sub-share' relation $\sqsubseteq$ is defined using Boolean union:

$$
a \sqsubseteq b ~\defi~ a \sqcup b = b .
$$

In this section, we will prove the following precise complexity of $\bmodel$:

\begin{theorem}\label{thm:forder}
The first-order theory of $\bmodel$ is $\le_{log}$-complete for $\mathsf{STA}(\ast,2^{n^{O(1)}},n)$, even if we allow arbitrary tree constants in the formulae.
\end{theorem}

One important implication of the above result is that the same complexity result still holds even if the additive operator $\oplus$ is included into the structure:

\begin{corollary}
The Boolean tree share structure with addition $\bmodel_{\oplus} = \struct{\mathbb{T},\oplus,\sqcup,\sqcap,\bar{\cdot}}$ is $\le_{log}$-complete for $\mathsf{STA}(\ast,2^{n^{O(1)}},n)$, even with arbitrary tree constants in the formulae.
\end{corollary}
\begin{proof}
Recall that $\oplus$ can be defined in term of $\sqcup$ and $\sqcap$ without additional quantifier variable:
$$
a \oplus b = c ~~\defi~~ a \sqcup b = c ~\wedge~ a \sqcap b = \circ .
$$

As a result, one can transform, in linear time, any additive constraint into Boolean constraint using the above definition. Hence the result follows. \qed
\end{proof}

Theorem~\ref{thm:forder} is stronger than the result in~\cite{le16:complex} which proved the same complexity but for restricted tree share constants in the formulae:

\begin{proposition}[\cite{le16:complex}]\label{prop:caba}
The first-order theory of $\bmodel$, where tree share constants are $\{\bullet,\circ\}$, is $\le_{log}$-complete for $\mathsf{STA}(\ast,2^{n^{O(1)}},n)$.
\end{proposition}

The hardness proof for lower bound of Theorem~\ref{thm:forder} is obtained directly from Prop.~\ref{prop:caba}. To show that the same complexity holds for upper bound, we construct an $O(n^2)$ algorithm $\mathtt{flatten}$ (Alg.~\ref{al:flatten}) that transforms arbitrary tree share formula into an equivalent tree share formula whose constants are $\{\bullet,\circ\}$:

\begin{lemma}\label{lm:flat}
Suppose $\mathtt{flatten}(\Phi) = \Phi'$. Then:
\begin{enumerate}
\item $\Phi'$ only contains $\{\bullet,\circ\}$ as constants.
\item $\Phi$ and $\Phi'$ have the same number of quantifier alternations.
\item $\Phi$ and $\Phi'$ are equivalent with respect to $\bmodel$.
\item $\mathtt{flatten}$ is $O(n^2)$. In particular, if the size of $\Phi$ is $n$ then $\Phi'$ has size $O(n^2)$.
\end{enumerate}
\end{lemma}

\noindent \textbf{Proof of Theorem \ref{thm:forder}.} The lower bound follows from Prop.~\ref{prop:caba}. By Lemma~\ref{lm:flat}, we can use $\texttt{flatten}$ in Alg.~\ref{al:flatten} to transform a tree formula $\Phi$ into an equivalent formula $\Phi'$ of size $O(n^2)$ that only contains $\{\bullet,\circ\}$ as constants and has the same number of quantifier alternations as in $\Phi$. By Prop.~\ref{prop:caba}, $\Phi'$ can be solved in $\mathsf{STA}(\ast,2^{n^{O(1)}},n)$. This proves the upper bound and thus the result follows. \qed \par

\algrenewcommand\textproc{}
\begin{algorithm}[t]
\begin{algorithmic}[1]
\caption{Flattening a Boolean tree share formula}\label{al:flatten}
\Function{\texttt{flatten}}{$\Phi$}
\Require{$\Phi$ is a Boolean tree sentence}
\Ensure{Return an equivalent formula of height zero}
\If {$\texttt{height}(\Phi) = 0$} \Return $\Phi$
\Else
\State let $s$ be the shape of $\Phi$
\For {each atomic formula $\Psi$ in $\Phi$: $t^1 = t^2$ or $t^1 \textsf{ op } t^2 = t^3$, $\textsf{op} \in \{\sqcup,\sqcap\}$ }
		\State $[t_1^i,\ldots t_n^i] \leftarrow$ \Call{\texttt{split}}{$t^i,s$} for $i=1 \ldots n$ \Comment{$n$ is the number of leaves in $s$}
		\State $\Psi_i \leftarrow t^i_1 = t^i_2$ or $t^i_1 \textsf{ op } t^i_2 = t^i_3$ for $i=1 \ldots n$
		\State $\Phi \leftarrow$ replace $\Psi$ with $\bigwedge_{i=1}^n \Psi_i$
	  \EndFor
	  \For {each quantifier $Qv$ in $\Phi$}
	  	\State $[v_1,\ldots,v_n] \leftarrow$ \Call{\texttt{split}}{$v,s$}
	  	\State $\Phi \leftarrow $ replace $Qv$ with $Qv_1\ldots Qv_n$
	  \EndFor
	\State \Return $\Phi$
\EndIf
\EndFunction

\State
\Function{\texttt{split}}{$t,s$}
\Require $t$ is either a variable or a constant, $s$ is a shape
\Ensure Return a list of decomposing components of $t$ according to shape $s$
\If {$s = \sleaf$} {\Return $[t]$}
\Else { let $s = \Tree [ $s_0$ $s_1$ ]$ in}
  \If {$t$ is $\bullet$ or $\circ$}
  \Return \listc{\Call{\texttt{split}}{t,s_0}}{\Call{\texttt{split}}{t,s_1}}
    \ElsIf { let $t = \Tree [ $t_1$ $t_2$ ]$ in}
  \Return \listc{\Call{\texttt{split}}{t_0,s_0}}{\Call{\texttt{split}}{t_1,s_1}}
  \Else {$t$ is a variable}
  \Return \listc{\Call{\texttt{split}}{t_0,s_0}}{\Call{\texttt{split}}{t_1,s_1}}
  \EndIf
\EndIf
\EndFunction
\end{algorithmic}
\end{algorithm}

It remains to prove the correctness of Lemma~\ref{lm:flat}. But first, we will provide a descriptive explanation for the control flow of $\texttt{flatten}$ in Alg.~\ref{al:flatten}. On line 2, it checks whether the height of $\Phi$, which is defined to be the height of the highest tree constant in $\Phi$, is zero.  If it is the case then no further computation is needed as $\Phi$ only contains $\{\bullet,\circ\}$ as constants. Otherwise, the shape $s$ (Definition~\ref{def:shape}) is computed on line 4 to guide the subsequent decompositions. On lines 5-9, each atomic sub-formula $\Psi$ is decomposed into sub-components according to the shape $s$ by the function $\texttt{split}$ described on lines 18-26. Intuitively, $\texttt{split}$ decomposes a tree $\tau$ into subtrees (line 21-22) or a variables $v$ into new variables with appropriate binary subscripts (line 23). On line 8, the formula $\Psi$ is replaced with the conjunction of its sub-components $\bigwedge_{i=1}^n \Psi_i$. Next, each quantifier variable $Qv$ in $\Phi$ is also replaced with a sequence of quantifier variables $Qv_1 \ldots Qv_n$ (lines 10-13). Finally, the modified formula $\Phi$ is returned as the result on line 14. The following example demonstrates the algorithm in action:

\begin{example}
Let $\Phi:~\forall a \exists b.~a \sqcup b = \Tree [ [ $\bullet$ $\circ$ ] $\circ$ ] \vee \neg (\bar{a} = \Tree [ $\circ$ [ $\bullet$ $\circ$ ] ])$. Then $\texttt{height}(\Phi) = 2 > 0$ and its shape $s$ is $\Tree [ [ $\sleaf$ $\sleaf$ ] [ $\sleaf$ $\sleaf$ ] ]$. Also, $\Phi$ contains the following atomic sub-formulae:
$$
\Psi:~ a \sqcup b = \Tree [ [ $\bullet$ $\circ$ ] $\circ$ ] ~~~\text{and}~~~
\Psi':~ \bar{a} = \Tree [ $\circ$ [ $\bullet$ $\circ$ ] ]~.
$$

After applying the $\texttt{split}$ function to $\Psi$ and $\Psi'$ with shape $s$, we acquire the following components:
\begin{enumerate}
\item $\Psi_1:~ a_{00} \sqcup b_{00} = \bullet$, $\Psi_2:~ a_{01} \sqcup b_{01} = \circ$, $\Psi_3:~ a_{10} \sqcup b_{10} = \circ$, $\Psi_4:~ a_{11} \sqcup b_{11} = \circ$.
\item $\Psi'_1:~ \overbar{a_{00}} = \circ$, $\Psi'_2:~ \overbar{a_{01}} = \circ$, $\Psi'_3:~ \overbar{a_{10}} = \bullet$, $\Psi'_4:~ \overbar{a_{11}} = \circ$.
\end{enumerate}

The following result formula is obtained by replacing $\Psi$ with $\bigwedge_{i=1}^4 \Psi_i$, $\Psi'$ with $\bigwedge_{i=1}^4 \Psi'_i$, $\forall a$ with $\forall a_{00}\forall a_{01}\forall a_{10}\forall a_{11}$, and $\exists b$ with $\exists b_{00}\exists b_{01}\exists b_{10}\exists b_{11}$:

$$
\forall a_{00}\forall a_{01}\forall a_{10}\forall a_{11} \exists b_{00}\exists b_{01}\exists b_{10}\exists b_{11}.~\bigwedge_{i=1}^4 \Psi_i \vee \neg (\bigwedge_{i=1}^4 \Psi'_i).
$$

\end{example}

\begin{definition}[Tree shape]\label{def:shape}
A shape of a tree $\tau$, denoted by $\shapet{\tau}$, is obtained by replacing its leaves with $\sleaf$, \emph{e.g.} $\shapet{\Tree [ $\bullet$ [ $\bullet$ $\circ$ ] ]} = \Tree [ $\sleaf$ [ $\sleaf$ $\sleaf$ ] ]$. The combined shape $s_1 \sqcup s_2$ is defined by overlapping $s_1$ and $s_2$, \emph{e.g.} $\Tree [ $\sleaf$ [ $\sleaf$ $\sleaf$ ] ] \sqcup \Tree [ [ $\sleaf$ $\sleaf$ ] $\sleaf$ ] =  \Tree [ [ $\sleaf$ $\sleaf$ ] [ $\sleaf$ $\sleaf$ ] ]$. The shape of a formula $\Phi$, denoted by $\shapet{\Phi}$, is the combined shape of its tree constants and $\sleaf$.
\end{definition}

Note that tree shapes are not canonical, otherwise all shapes are collapsed into a single shape $\sleaf$. We are now ready to prove the first three claims of Lemma~\ref{lm:flat}:
\par
\noindent \textbf{Proof of Lemma~\ref{lm:flat}.1, \ref{lm:flat}.2 and \ref{lm:flat}.3}. Observe that the shape of each atomic sub-formula $\Psi$ is `smaller' than the shape of $\Phi$, \emph{i.e.} $\shapet{\Psi} \sqcup \shapet{\Phi} = \shapet{\Phi}$. As a result, each formula in the decomposition of $\texttt{split}(\Psi,\shapet{\Phi})$ always has height zero, \emph{i.e.} its only constants are $\{\bullet,\circ\}$. This proves claim 1.

Next, recall that the number of quantifier alternations is the number of times where quantifiers are switched from $\forall$ to $\exists$ or vice versa. The only place that $\texttt{flatten}$ modifies quantifiers is on line 12 in which the invariant for quantifier alternations is preserved. As a result, claim 2 is also justified.

We are left with the claim that $\texttt{flatten}$ is $O(n^2)$ where $n$ is the size of the input formula $\Phi$. By a simple analysis of $\texttt{flatten}$, it is essentially equivalent to show that the result formula has size $O(n^2)$. First, observe that the formula shape $\shapet{\Phi}$ has size $O(n)$ and thus we need $O(n)$ decompositions for each atomic sub-formula $\Psi$ and each quantifier variable $Qv$ of $\Phi$. Also, each component in the decomposition of $\Psi$ (or $Qv$) has size at most the size of $\Psi$ (or $Qv$). As a result, the size of the formula $\Phi'$ only increases by a factor of $O(n)$ compared to the size of $\Phi$. Hence $\Phi'$ has size $O(n^2)$. \qed

To prove claim 4, we first establish the following result about the $\texttt{split}$ function. Intuitively, this lemma asserts that one can use $\texttt{split}$ together with some tree shape $s$ to construct an isomorphic Boolean structure whose elements are lists of tree shares:

\begin{lemma}\label{lm:split}

Let \emph{$\texttt{split}_s$}$\defi \lambda \tau.$ \emph{\texttt{split}}$(\tau,s)$, \emph{e.g.} \emph{$\texttt{split}$}$_{\Tree [ $\sleaf$ [ $\sleaf$ $\sleaf$ ] ]}(\Tree [ [ $\bullet$ $\circ$ ] $\bullet$ ]) = [\Tree [ $\bullet$ $\circ$ ], \bullet,\bullet]$. Then \emph{$\texttt{split}_s$} is an isomorphism from $\bmodel$ to $\bmodel' = \struct{\mathbb{T}^n,\sqcup',\sqcap',\bar{\cdot}'}$ where $n$ is the number of leaves in $s$ and each operator in $\mathcal{M}'$ is defined component-wise from the corresponding operator in $\bmodel$, \emph{e.g.} $[a_1,a_2] \sqcup' [b_1,b_2] = [a_1 \sqcup a_2,b_1 \sqcup b_2]$.
\end{lemma}
\begin{proof}
W.l.o.g. we will only prove the case $s = \Tree [ $\sleaf$ $\sleaf$ ]$ as similar argument can be obtained for the general case. By inductive arguments, we can prove that $\texttt{split}_s$ is a bijection from $\mathbb{T}$ to $\mathbb{T} \times \mathbb{T}$. Furthermore:
\begin{enumerate}
\item $\texttt{split}_s(a)~\diamond~\texttt{split}_s(b) = \texttt{split}_s(c)$ iff $a~\diamond~b = c$ for $\diamond \in \{\sqcup,\sqcap\}$.
\item $\texttt{split}_s(\bar{\tau}) = \overbar{\texttt{split}_s(\tau)}$ .
\end{enumerate}

Hence $\mathsf{split}_s$ is an isomorphism from $\bmodel$ to $\bmodel' = \struct{\mathbb{T}\times\mathbb{T},\sqcup',\sqcap',\bar{\cdot}'}$. \qed
\end{proof}

\noindent{\textbf{Proof of Lemma~\ref{lm:flat}.4}}. By Lemma~\ref{lm:split}, the function $\texttt{split}_s$ allows us to transform formulae in $\bmodel$ into equivalent formulaes over tree share lists in $\bmodel' = \struct{\mathbb{T}^n,\sqcup',\sqcap',\bar{\cdot}'}$. On the other hand, observe that formulae in $\bmodel'$ can be rewritten into equivalent formulae in $\bmodel$ using conjunctions and extra quantifier variables, \emph{e.g.} $\exists a \forall b.~a \sqcup' b = [\Tree [ $\circ$ $\bullet$ ], \bullet]$ is equivalent to $\exists a_1 \exists a_2 \forall b_1 \forall b_2.~ a_1 \sqcup b_1 = \Tree [ $\circ$ $\bullet$ ] \wedge a_2 \sqcup b_2 = \bullet$.  Hence the result follows. \qed

The correctness of Lemma~\ref{lm:flat} is now fully justified. We end this section by pointing out a refined complexity result for the existential theory of $\bmodel$, which corresponds to the satisfiability problem of quantifier-free formulae. Note that the number of quantifier alternations for this fragment is zero, and thus Theorem~\ref{thm:forder} only gives us an upper bound $\mathsf{STA}(\ast,2^{n^{O(1)}},0)$, which is exponential time complexity. Instead, we can use Lemma~\ref{lm:flat} to acquire the precise complexity:

\begin{corollary}
The existential theory of $\bmodel$, with arbitrary tree share constants, is $\mathsf{NP}$-complete.
\end{corollary}
\begin{proof}
Recall a classic result that existential theory of Countably Atomless BAs is $\mathsf{NP}$-complete~\cite{kim96:nba}. As $\bmodel$ belongs to this class, the lower bound is justified. To see why the upper bound holds, we use the function $\texttt{flatten}$ to transform the input formula into standard BA formula and thus the result follows from Lemma~\ref{lm:flat}. \qed
\end{proof}

\section{Complexity of combined structure $\lmodel \defi \struct{\treedom,\sqcup,\sqcap,\bar{\cdot},\rbow{\tau}}$}
\label{sec:combresult}
In addition to the Boolean operators in \S\ref{sec:boolean}, recall from
\S\ref{sub:overview} that tree shares also
possess a multiplicative operator $\bowtie$ that resembles the multiplication
of rational permissions. As mentioned in \S\ref{sub:treepre},
\cite{le16:complex}~showed that $\bowtie$ is isomorphic to string concatenation,
implying that the first-order theory of $ \struct{\treedom,\bowtie}$ is undecidable,
and so of course the first-order theory of $ \struct{\treedom,\sqcup,\sqcap,\bar{\cdot},\bowtie} $
is likewise undecidable.

By restricting multiplication to have only constants on
the right-hand side, however, \emph{i.e.} to the family of unary
operators $\rbow{\tau}(x) ~\defi~ x \bowtie \tau$, Le \emph{et al.} showed that
decidability of the first-order theory was restored
for the combined structure $\lmodel \defi \struct{\treedom,\sqcup,\sqcap,\bar{\cdot},\rbow{\tau}}$.
However, Le \emph{et al.} were not able to specify any particular complexity
class.  In this section, we fill in this blank by proving that the
first-order theory of $\lmodel$ is nonelementary, \emph{i.e.} that it cannot be
solved by any resource-bound (space or time) algorithm:
\begin{theorem}\label{thm:lmodel}
The first-order theory of $\lmodel$ is non-elementary.
\end{theorem}

\def \ltr {\ensuremath{\mathcal{L}}}
\def \rtr {\ensuremath{\mathcal{R}}}

\hide{

Having analyzed the complexity proof of the Boolean
structure~\S\ref{sec:boolean} and multiplicative
structure~\S\ref{sec:bowresult}, we are now interested in finding the
complexity of the combined structure that utilizes both Boolean and
multiplicative operators. Since $\bowtie$ is isomorphic to string
concatenation, any structure includes this operator has its first-order
theory undecidable, unless certain restrictions are enforced. One of the
results in~\cite{le16:complex} stated that, by restricting $\bowtie$ to
always have constants on the right, we recover the first-order decidability
for theory of the combined structure

, yet the exact
complexity remains unknown.
}

To prove Theorem~\ref{thm:lmodel}, we reduce the binary string structure with
prefix relation~\cite{COMPTON19901}, which is known to be nonelementary, into
$\lmodel$. Here we recall the definition and complexity result of binary
strings structure:

\begin{proposition}[\cite{COMPTON19901,meyer:thesis}]\label{prop:nonele}
Let $\tmodel = \struct{\{0,1\}^{*},S_0,S_1, \preceq}$ be the binary string structure in which $\{0,1\}^{*}$ is the set of binary strings, $S_i$ is the successor function s.t. $S_i(s) = s \cdot i$, and $\preceq$ is the binary prefix relation s.t. $x \preceq y$ iff there exists $z$ satisfies $x \cdot z = y$.
Then the first-order theory of $\tmodel$ is non-elementary.
\end{proposition}

Before going into the technical detail, we briefly explain the many-one
reduction from $\tmodel$ into $\lmodel$.  The key idea is that the set of
binary strings $\{0,1\}^*$ can be bijectively mapped into the set of
\emph{unary trees} $\utree$, trees that have exactly one black leaf,
\emph{e.g.} $\{\bullet, \Tree [ $\bullet$ $\circ$ ], \Tree [ $\circ$
$\bullet$ ], \Tree[ $\circ$ [ $\bullet$ $\circ$ ] ], \cdots \}$. For
convenience, we use the symbol $\ltr$ to represent the left tree $\ltree$ and
$\rtr$ for the right tree $\rtree$. Then:

\begin{lemma}\label{lm:domain}
Let $\embed$ map $\struct{\{0,1\}^{*},S_0,S_1, \preceq}$ into $ \struct{\treedom,\sqcup,\sqcap,\bar{\cdot},\rbow{\tau}}$ such that:
\begin{enumerate}
\item $\embed(\epsilon) = \bullet$, $\embed(0) = \ltr$, $\embed(1) = \rtr$.
\item $\embed(b_1\ldots b_n) = \embed(b_1)\bowtie \ldots \bowtie \embed(b_n), b_i \in \{0,1\}$.
\item $g(S_0) = \lambda s.\bowtie_\ltr(g(s))$, $g(S_1) = \lambda s.\bowtie_\rtr(g(s))$.
\item $\embed(x \preceq y) = \embed(y) \sqsubseteq \embed(x)$ where $\tau_1 \sqsubseteq \tau_2 \defi \tau_1 \sqcup \tau_2 = \tau_2$.
\end{enumerate}

Then $\embed$ is a bijection from $\{0,1\}^{*}$ to $\utree$, and $x \preceq
y$ iff $\embed(y) \sqsubseteq \embed(x)$.
\end{lemma}
\begin{proof}
The routine proof that $\embed$ is bijective is done by induction on the
string length. Intuitively, the binary string $s$ corresponds to the path
from the tree root in $\embed(s)$ to its single black leaf, where $0$ means
`go left' and $1$ means `go right'. For example, the tree $\embed(110) = \rtr
\bowtie \rtr \bowtie \ltr = \rtree \bowtie \rtree \bowtie \ltree = \Tree [
$\circ$ [ $\circ$ [ $\bullet$ $\circ$ ] ] ]$ corresponds to the path
right$\rightarrow$right$\rightarrow$left.

Now observe that if $\tau_1,\tau_2$ are unary trees then $\tau_1 \sqsubseteq
\tau_2$ (\emph{i.e.} the black-leaf path in $\tau_2$ is a sub-path of the
black-leaf path in $\tau_1$) iff there exists a unary tree $\tau_3$ such that
$\tau_2 \bowtie \tau_3 = \tau_1$ (intuitively, $\tau_3$ represents the
difference path between $\tau_2$ and $\tau_1$). Thus $x \preceq y$ iff there
exists $z$ such that $~xz = y$, iff $\embed(x) \bowtie \embed(z) =
\embed(y)$, which is equivalent to $\embed(y) \sqsubseteq \embed(x)$ by the
above observation. \qed
\end{proof}

In order for the reduction to work, we need to express the type of unary
trees using operators from $\lmodel$. The below lemma shows that the type of
$\utree$ is expressible via a universal formula in $\lmodel$:

\begin{lemma}\label{lm:express}
A tree $\tau$ is unary iff it satisfies the following $\forall$-formula:

$$
\tau \neq \circ ~\wedge~ \big(\forall \tau'.~\tau' \bowtie \ltr \sqsubset \tau ~\leftrightarrow~ \tau' \bowtie \rtr \sqsubset \tau\big).
$$

where $\tau_1 \sqsubset \tau_2 \defi \tau_1 \sqcup \tau_2 = \tau_2 \wedge \tau_1 \neq \tau_2$.
\end{lemma}
\begin{proof}

The $\Rightarrow$ direction is proved by induction on the height of $\tau$.
The key observation is that if $\tau_1 \bowtie \tau_2 \sqsubset \tau_3$ and
$\tau_2,\tau_3$ are unary then $\tau_1$ is also unary, $\tau_1 \sqsubseteq
\tau_3$ and thus $\tau_1 \bowtie \overbar{\tau_2} \sqsubset \tau_3$. Note
that both $\ltr, \rtr$ are unary and $\overbar{\ltr} = \rtr$, hence the
result follows.

For $\Leftarrow$, assume $\tau$ is not unary. As $\tau \neq \circ$, it
follows that $\tau$ contains at least two black leaves in its representation.
Let $\tau_1$ be the tree that represents the path to one of the black leaves
in $\tau$, we have $\tau_1 \sqsubset \tau$ and for any unary tree $\tau_2$,
if $\tau_1 \sqsubset \tau_2$ then $\tau_2 \not \sqsubseteq \tau$. As $\tau_1$
is unary, we can rewrite $\tau_1$ as either $\tau'_1 \bowtie \ltr$ or
$\tau'_1 \bowtie \rtr$ for some unary tree $\tau'_1$. The latter disjunction
together with the equivalence in the premise give us both $\tau'_1 \bowtie
\ltr \sqsubset \tau$ and $\tau'_1 \bowtie \rtr \sqsubset \tau$. Also, we have
$\tau_1 \sqsubset \tau'_1$ and thus $\tau'_1 \not \sqsubseteq \tau$ by the
aforementioned observation. Hence $\tau'_1 = \tau_1 \bowtie \bullet = \tau'_1
\bowtie (\ltr \sqcup \rtr) \sqsubseteq \tau$ which is a contradiction. \qed
\end{proof}

\noindent \textbf{Proof of Theorem~\ref{thm:lmodel}.} We employ the reduction
technique in~\cite{Gradel1990235} where formulae in $\tmodel$ are interpreted
using the operators from $\lmodel$. The interpretation of constants and
operators is previously mentioned and justified in Lemma~\ref{lm:domain}. We
then replace each sub-formula $\exists x.~\Phi$ with $\exists x.~x \in \utree
\wedge \Phi$ and $\forall x.~\Phi$ with $\forall x.~x \in \utree \rightarrow
\Phi$ using the formula in Lemma~\ref{lm:express}. It follows that the
first-order complexity of $\lmodel$ is bounded below by the first-order
complexity of $\tmodel$. Hence by Prop.~\ref{prop:nonele}, the first-order
complexity of $\lmodel$ is nonelementary. \qed

\section{Causes of, and mitigants to, the nonelementary bound}
\label{sec:bowresult}
Having proven the nonelementary lower bound for the combined theory in \S\ref{sec:combresult},
we discuss causes and mitigants.  In \S\ref{subsec:complexmult} we show that the
nonelementary behavior of $\lmodel$ comes from the combination of both the additive
and multiplicative theories by proving an elementary upper bound on a generalization
of the multiplicative theory, and in \S\ref{subsec:combformpractice} we discuss
why we believe that verification tools in practice will avoid the nonelementary lower bound.

\subsection{Complexity of multiplicative structure $\brmodel \defi \struct{\treedom,\lbow{\tau},\rbow{\tau}}$}
\label{subsec:complexmult}

Since the first-order theory over $ \struct{\treedom, \bowtie} $ is
undecidable, it may seem plausible that the nonelementary behaviour of
$\lmodel$ comes from the $\rbow{\tau}$ subtheory rather than the ``simpler''
Boolean subtheory $\bmodel$, even though the specific proof of the lower bound
given in \S\ref{sec:combresult} used both the additive and multiplicative
theories (\emph{e.g.} in Lemma~\ref{lm:express}). This intuition, however, is mistaken.  In
fact, even if we generalize the theory to allow multiplication by constants
on either side---\emph{i.e.}, by adding $\lbow{\tau\,}(x) \defi \tau \bowtie
x$ to the language---the restricted multiplicative theory $\brmodel \defi
\struct{\treedom,\lbow{\tau},\rbow{\tau}}$ is elementary. Specifically, we
will prove that the first-order theory of $\brmodel$ is
$\mathsf{STA}(\ast,2^{O(n)},n)$-complete and thus elementarily decidable:
\begin{theorem}\label{thm:rmodel}
The first-order theory of $\brmodel$ is $\le_{\text{log-lin}}$-complete for $\mathsf{STA}(\ast,2^{O(n)},n)$.
\end{theorem}
Therefore, the nonelementary behavior of $\lmodel$ arises precisely because of
the combination of both the additive and multiplicative subtheories.

%take only constants as one operand, obtaining the two families of unary operators indexed by
%constants $\tau$:
%$$
%\lbow{\tau\,}(x) ~\defi~ \tau \bowtie x ~~~\text{and}~~~ \rbow{\tau}(x) ~\defi~ x \bowtie \tau.
%$$

We prove Theorem~\ref{thm:rmodel} by solving a similar problem in which two
tree shares $\{\bullet,\circ\}$ are excluded from the tree domain $\treedom$.
That is, let $\treedom^+ = \treedom \backslash \{\bullet,\circ\}$ and
$\brmodel^+ = \struct{\treedom^+,\lbow{\tau},\rbow{\tau}}$, we want:

\begin{lemma}\label{lm:rmodel}
The complexity of $\fth{\brmodel^+}$ is $\le_{\text{log-lin}}$-complete for $\mathsf{STA}(\ast,2^{O(n)},n)$.
\end{lemma}

By using Lemma~\ref{lm:rmodel}, the proof for the main theorem is straightforward:

\noindent \textbf{Proof of Theorem~\ref{thm:rmodel}}. The hardness proof is
direct from the fact that membership constraint in $\brmodel^{+}$ can be
expressed using membership constraint in $\brmodel$:

$$
\tau \in \brmodel^{+} ~~\text{iff}~~ \tau \in \brmodel \wedge \tau \neq \circ \wedge \tau \neq \bullet .
$$

As a result, any sentence from $\brmodel^{+}$ can be transformed into equivalent sentence in $\brmodel$ by rewriting each $\forall v. \Phi$ with $\forall v.(v \neq \circ \wedge v \neq \bullet) \rightarrow \Phi$ and each $\exists v. \Phi$ with $\exists v. v \neq \circ \wedge v \neq \bullet \wedge \Phi$.

To prove the upper bound, we use the guessing technique as in~\cite{le16:complex}. In detail, we partition the domain $\treedom$ into three disjoint sets:
 $$
 S_1 = \{\circ\} \quad\quad\quad S_2 = \{\bullet\} \quad\quad\quad S_3 = \treedom^+ .
 $$

 Suppose the input formula contains $n$ variables, we then use a ternary vector of length $n$ to guess the partition domain of these variables, \emph{e.g.}, if a variable $v$ is guessed with the value $i\in \{1,2,3 \}$ then $v$ is assigned to the domain $S_i$. In particular, if $v$ is assigned to $S_1$ or $S_2$, we substitute $v$ for $\circ$ or $\bullet$ respectively. Next, each bowtie term $\rbow{\tau}(a)$ or $\lbow{\tau\,}(a)$ that contains tree share constants $\bullet$ or $\circ$ is simplified using the following identities:
$$
\tau \bowtie \bullet ~=~ \bullet \bowtie \tau ~=~ \tau \quad\quad\quad
\tau \bowtie \circ ~=~ \circ \bowtie \tau ~=~ \circ.
$$

After this step, all the atomic sub-formulae that contain $\circ$ or $\bullet$ are reduced into either variable equalities $v_1 = v_2, v = \tau$ or trivial constant equalities such as $ \bullet = \bullet, \Tree [ $\bullet$ $\circ$ ] = \circ$ that can be replaced by either $\top$ or $\bot$. As a result, the new equivalent formula is free of tree share constants $\{\bullet,\circ\}$ whilst all variables are quantified over the domain $\treedom^+$. Such formula can be solved using the Turing machine that decides $\fth{\brmodel^+}$. The whole guessing process can be integrated into the alternating Turing machine without increasing the formula size or number of quantifiers (\emph{i.e.} the alternating Turing machine only needs to make two extra guesses $\bullet$ and $\circ$ for each variable and the simplification only takes linear time). Hence this justifies the upper bound. \qed

The rest of this section is dedicated to the proof of Lemma~\ref{lm:rmodel}. To prove the complexity $\fth{\brmodel^{+}}$, we construct an efficient isomorphism from $\brmodel^+$ to the structure of ternary strings in $\{0,1,2\}^\ast$ with prefix and suffix successors. The existence of such isomorphism will ensure the complexity matching between the tree structure and the string structure. Here we recall a result from~\cite{Rybina2003} about the first-order complexity of the string structure with successors:
\begin{proposition}[\cite{Rybina2003}]
Let $\mathcal{S} = \struct{\{0,1\}^\ast,P_0,P_1,S_0,S_1}$ be the structure of binary strings with prefix successors $P_0,P_1$ and suffix successors $S_0,S_1$ such that:
$$
P_0(s) = 0 \cdot s \quad\quad P_1(s) = 1\cdot s \quad\quad S_0(s) = s\cdot 0 \quad\quad S_1(s) = s\cdot 1.
$$

Then the first-order theory of $\mathcal{S}$ is $\le_{\text{log-lin}}$-complete for $\mathsf{STA}(\ast,2^{O(n)},n)$.
\end{proposition}

The above result cannot be used immediately to prove our main theorem. Instead, we use it to infer a more general result where successors are not only restricted to $0$ and $1$, but also allowed to be any string $s$ in a finite alphabet:

\begin{lemma}\label{lm:queue}
Let $\Sigma$ be a finite alphabet of size $k \geq 2$ and $\mathcal{S}' = \struct{\Sigma^\ast,P_s,S_s}$ the structure of $k$-ary strings with infinitely many prefix successors $P_s$ and suffix successors $S_s$ where $s \in \Sigma^\ast$ such that:
$$
P_s(s') = s \cdot s' \quad\quad\quad\quad S_s(s') = s' \cdot s.
$$

Then the first-order theory of $\mathcal{S}'$ is $\le_{\text{log-lin}}$-complete for $\mathsf{STA}(\ast,2^{O(n)},n)$.
\end{lemma}
\begin{proof}
Although the proof in~\cite{Rybina2003} only considers binary alphabet, the same result still holds even for finite alphabet $\Sigma$ of size $k \geq 2$ with $k$ prefix and suffix successors. Let $s = a_1\ldots a_n$ where $a_i \in \Sigma$, the successors $P_s$ and $S_s$ can be defined in linear size from successors in $\mathcal{S}$ as follows:
$$
P_s ~\defi~ \lambda s'.~ P_{a_1}(\ldots P_{a_n}(s')) \quad\quad\quad S_s ~\defi~ \lambda s'.~ S_{a_n}(\ldots S_{a_1}(s')).
$$
These definitions are quantifier-free and thus the result follows. \qed
\end{proof}

Next, we recall some key results from~\cite{le16:complex} that establishes the fundamental connection between trees and strings in word equation:

\begin{proposition}[\cite{le16:complex}]\label{prop:uniquetree} We call a tree $\tau$ in $\treedom^+$ prime if $\tau = \tau_1 \bowtie \tau_2$ implies either $\tau_1 = \bullet$ or $\tau_2 = \bullet$. Then for each tree $\tau$ in $\treedom^+$, there exists a unique sequence of prime trees $\{\tau_i\}_{i=1}^n$ such that $\tau = \tau_1 \bowtie \cdots \bowtie \tau_n$. As a result, each tree in $\treedom^+$ can be treated as a string in a word equation in which the alphabet is $\mathbb{P}$, the countably infinite set of prime trees, and $\bowtie$ is the string concatenation.
\end{proposition}

For example, the factorization of $\Tree [ [ [ $\circ$ $\bullet$ ] $\circ$ ] [ $\circ$ [ [ $\circ$ $\bullet$ ] $\circ$ ] ] ]$ is $\Tree [ $\bullet$  [ $\circ$ $\bullet$ ] ] \bowtie \Tree [ $\bullet$ $\circ$ ] \bowtie \Tree [ $\circ$ $\bullet$ ]$, which is unique. Prop.~\ref{prop:uniquetree} asserts that by factorizing tree shares into prime trees, we can effectively transform multiplicative tree share constraints into equivalent word equations. Ideally, if we can represent each prime tree as a unique letter in the alphabet then Lemma~\ref{lm:rmodel} would follow from Lemma~\ref{lm:queue}. Unfortunately, the set of prime trees $\mathbb{P}$ are infinite~{\cite{le16:complex} while Lemma~\ref{lm:queue} requires a finite alphabet. As a result, our tree encoding needs to be more sophisticated than the na\"{i}ve way. The key observation here is that, as $\mathbb{P}$ is countably infinite, there must be a bijective \emph{encoding function} $I:\mathbb{P} \mapsto \{0,1\}^\ast$ that encodes each prime tree into binary string, including the empty string $\epsilon$. We need not to know the construction of $I$ in advance, but it is important to keep in mind that $I$ exists and the delay of its construction is intentional. We then extend $I$ into $\hat{I}$ that maps tree shares in $\treedom^+$ into ternary string in $\{0,1,2\}^*$ where the letter $2$ purposely represents the delimiter between two consecutive prime trees:

\begin{lemma}\label{lm:bi}
Let $\hat{I}: \treedom^+ \mapsto \{0,1,2\}^*$ be the mapping  from tree shares into ternary strings such that for prime trees $\tau_i \in \mathbb{P}$ where $i \in \{1,\ldots,n\}$, we have:

$$
\hat{I}(\tau_1 \bowtie \ldots \bowtie \tau_n) =  I(\tau_1)\cdot 2\ldots 2 \cdot I(\tau_n).
$$

By Prop.~\ref{prop:uniquetree}, $\hat{I}$ is bijective. Furthermore, let $\tau_1,\tau_2 \in \treedom^+$ then:
$$
\hat{I}(\tau_1 \bowtie \tau_2) = \hat{I}(\tau_1) \cdot 2 \cdot \hat{I}(\tau_2) .
$$
\end{lemma}

Having the core encoding function $\hat{I}$ defined, it is now routine to establish the isomorphism from the tree structure $\brmodel^{+}$ to the string structure $\mathcal{S}'$:

\begin{lemma}\label{lm:iso}
Let $f$ be a function that maps the tree structure $\struct{\treedom^+,_{\tau}{\bowtie},\bowtie_{\tau}}$ into the string structure $\struct{\{0,1,2\},P_{s2},S_{2s}}$ such that:
\begin{enumerate}
\item For each tree $\tau \in \treedom^+$, we let $f(\tau) ~\defi~ \hat{I}(\tau)$.
\item For each function $_{\tau}{\bowtie}$, we let $f(_{\tau}{\bowtie}) ~\defi~ P_{\hat{I}(\tau)2}$.
\item For each function $\rbow{\tau}$, we let $f(\rbow{\tau}) ~\defi~ S_{2\hat{I}(\tau)}$.
\end{enumerate}

Then $f$ is an isomorphism from $\brmodel^{+}$ to $\mathcal{S}'$.
\end{lemma}
\textbf{Proof of Lemma~\ref{lm:rmodel}}.
For the upper bound, observe that the function $f$ in Lemma~\ref{lm:iso} can be used to transform tree share formulae in $\brmodel^{+}$ to string formulae in $\mathcal{S}'$. It remains to ensure that the size of the string formula is not exponentially exploded. In particular, it suffices to construct $\hat{I}$ such that if a tree $\tau \in \treedom^+$ has size $n$, its corresponding string $\hat{I}(\tau)$ has linear size $O(n)$. Recall that $\hat{I}$ is extended from $I$ which can be constructed in many different ways. Thus to avoid the size explosion, we choose to specify the encoding function $I$ on the fly \emph{after observing the input tree share formula}. To be precise, given a formula $\Phi$ in $\brmodel$, we first factorize all its tree constants into prime trees, which can be done in log-space~\cite{le16:complex}. Suppose the formula has $n$ prime trees $\{\tau_i\}_{i=1}^{n}$ sorted in the ascending order of their sizes, we choose the most efficient binary encoding by letting $I(\tau_i) = s_i$ where $s_i$ is the $i^{\text{th}}$ string in length-lexicographic (shortlex) order of $\{0,1\}^\ast$ ,\emph{i.e.} $\{\epsilon,0,1,00,01,\ldots\}$. This encoding ensures that the size of $\tau_i$ and the length of $s_i$ only differ by a constant factor. Given the fact that a tree share in its factorized form $\tau_1 \bowtie \ldots \bowtie \tau_n$ only requires $O(\sum_{i=1}^n\hat{I}(\tau_i))$ bits to represent, we infer that its size and the length of its string counterpart $\hat{I}(\tau)$ also differ by a constant factor. Hence, the upper bound complexity is justified.

To prove the lower bound, we need to construct the inverse function $f^{-1}$ that maps the string structure $\mathcal{S}'$ into the tree share structure $\brmodel$. Although the existence of $f^{-1}$ is guaranteed since $f$ is isomorphism, we also need to take care of the size explosion problem. It boils down to construct an efficient mapping $I^{-1}$ from binary strings to prime trees by observing the input string formula $\Phi$. For each string constant $s_12\ldots 2s_n$ in $\Phi$ where $s_i \in \{0,1\}^*$, we extract all of the binary strings $s_i$. We then maps each distinct binary string $s_i$ to a unique prime tree $\tau_i$ as follows. Let $k(0) = \Tree [ $\bullet$ $\circ$ ]$, $k(1) = \Tree [ $\circ$ $\bullet$ ]$ and assume $s_i = a_0 \ldots a_m$ for $a_i \in \{0,1\}$, we compute $\tau = k(a_0) \bowtie \ldots \bowtie k(a_m)$. Then the mapped tree share for the string $s_i$ is constructed as $\tau_i = \Tree [ $\bullet$ $\tau$ ]$ (if $s_i = \epsilon$ then $\tau_i = \Tree [ $\bullet$ $\circ$ ]$ ). It follows that $\tau_i$ is prime and this skewed tree has size $O(n)$ where $n$ is the length of $s_i$. Thus the result follows. \qed

\begin{example}
Consider the tree formula $\forall a \exists b \exists c.~a = b \bowtie \Tree [ [ $\circ$ $\bullet$ ] $\circ$ ] \wedge b = \Tree [ $\circ$ [ $\bullet$ $\circ$ ] ] \bowtie c $. This formula contains two constants whose factorizations are below:
$$
c_1 = \Tree [ [ $\circ$ $\bullet$ ] $\circ$ ] ~=~ \ltree ~\bowtie~ \rtree \quad\quad\quad\quad
c_2 = \Tree [ $\circ$ [ $\bullet$ $\circ$ ] ] ~=~ \rtree ~\bowtie~ \ltree~.
$$

We choose $I$ such that $I(\ltree) = \epsilon$ and $I(\rtree) = 0$. Our encoding gives $s_1 = 20$ and $s_2 = 02$. This results in the string formula $ \forall a \exists b \exists c.~a = S_{220}(b) \wedge b = P_{022}(c)$ whose explicit form is $ \forall a \exists b \exists c.~a = b220 \wedge b = 022c$.

Now suppose that we want to transform the above string formula into equivalent tree formula. Following the proof of Lemma~\ref{lm:rmodel}, we extract from the formula two binary strings $s_1 = \epsilon$ and $s_2 = 0$ which are mapped to the prime trees $\tau_1 = \Tree [ $\bullet$ $\circ$ ]$ and $ \tau_2 = \Tree [ $\bullet$ [ $\bullet$ $\circ$ ] ]$ respectively. Hence the equivalent tree share formula is $ \forall a \exists b \exists c. a= \bowtie_{\tau_1 \bowtie \tau_2}(b) \wedge b = _{\tau_2 \bowtie \tau_1}\bowtie(c)$. It is worth noticing the difference between this tree formula and the original tree formula, which suggests the fact that the representation of the alphabet (\emph{i.e.} prime trees) is not important.
\end{example}

\subsection{Combined $\lmodel$ formulae in practice}
\label{subsec:combformpractice}

The source of the nonelementary behavior comes from two factors. First, as
proven just above, it comes from the combination of both the additive and
multiplicative operations of tree shares.  Second, it comes from the number
of quantifier alternations in the formula being analyzed, due to the
encoding of $\lmodel$ in tree automata~\cite{le16:complex} and the resulting
upper bound (the transformed automata of first-order formulae of tree automatic
structures have sizes bounded by a tower of exponentials whose height is the
number of quantifier alternations~\cite{blumensath04:as,blumensath99:as}).

Happily, in typical verifications, especially in highly-automated verifications
such as those done by tools like HIP/SLEEK~\cite{le17:certproc}, the number
of quantifier alternations in formulae is small, even when carrying out complex
verifications or inference.  For example, consider the following biabduction
problem (a separation-logic-based inference procedure) handled by the ShareInfer
tool from~\cite{le17:logic}:
\[
a \stackrel{\pi}{\mapsto} (b,c,d) \star \Tree [ $\bullet$ $\circ$ ] \cdot \pi \cdot \mathsf{tree}(c) \star \null
\Tree [ $\circ$ $\bullet$ ] \cdot \pi \cdot \mathsf{tree}(d) \star [??] \vdash \null
\Tree [ $\bullet$ $\circ$ ] \cdot \pi \cdot \mathsf{tree}(a) \star [??]
\]
ShareInfer will calculate $\Tree [ $\bullet$ $\circ$ ] \cdot \pi \cdot \mathsf{tree}(d)$
for the antiframe and
$a \xmapsto{\pi ~ \bowtie ~ \Tree [ $\circ$ $\bullet$ ] } (b,c,d) \star \null
\Tree [ $\circ$ $\bullet$ ] \cdot \pi \cdot \mathsf{tree}(d)$ for the inference
frame.  Although these guesses are a bit sophisticated,
verifing them depends on~\cite{cris:thesis} the following
quantifier-alternation-free $\lmodel$ sentence: $\forall \pi, \pi'.~ \pi = \pi' ~ \Rightarrow
~ \rbow{\Tree [ $\bullet$ $\circ$ ] } \! (\pi) \oplus \rbow{\Tree [ $\bullet$ $\circ$ ] } \! (\pi) = \pi'$.
Even with complex loop invariants, more than one alternation would be surprising
because \emph{e.g.} verification tools tend to maintain formulae in well-chosen canonical forms.

Moreover, because tree automata are closely connected to other well-studied domains,
we can take advantage of existing tools such as MONA~\cite{monamanual2001}.
As an experiment we have hand-translated $\lmodel$ formulae
into WS2S, the language of MONA, using the techniques of~\cite{DBLP:journals/lmcs/ColcombetL07}. The technical details of the translation are provided in appendix~\S\ref{sec:formula}. 
For the above formula, MONA reported 205 DAG hits and 145 nodes, with essentially a 0ms running time.

Lastly, heuristics are well-justified both because of the restricted
problem formats we expect in practice as well as because of the nonelementary
worst-case lower bound we proved in \S\ref{sec:combresult}, opening the door
to newer techniques like antichain/simulation~\cite{simulation}.

\section{Future work and conclusion}
\label{sec:conclude}
We have developed a tighter understanding of the complexity of the tree share
model.  As Boolean Algebras, their first-order theory is
$\mathsf{STA}(\ast,2^{n^{O(1)}},n)$-complete, even with arbitrary tree
constants in the formulas.  Although the first-order theory over tree
multiplication is undecidable~\cite{le16:complex}, we have found that by
restricting multiplication to be by a constant (on both the left
$\lbow{\tau}$ and right $\rbow{\tau}$ sides) we obtain a substructure
$\brmodel$ whose first-order theory is
$\mathsf{STA}(\ast,2^{O(n)},n)$-complete.  Accordingly, we have two
structures whose first-order theory has elementary complexity.
Interestingly, their combined theory is still decidable but nonelementary, even if we only allow
multiplication by a constant on the right $\rbow{\tau}$.

We have several directions for future work. It is natural to investigate the
precise complexity of the existential theory with the Boolean operators
and right-sided multiplication $\rbow{\tau}$ (structure $\lmodel$). The encoding
into tree-automatic structures from \cite{le16:complex} provides only an
exponential-time upper bound (because of the result for the corresponding
fragment in tree-automatic structures, e.g., see \cite{anthony-thesis}), and 
there is the obvious NP lower bound that comes
from propositional logic satisfiability.
 %since our
%encoding of $\utree$ in Lemma~\ref{lm:express} uses universals.
%We do not know if the existential theory over the restricted multiplication (structure $\brmodel$) has a better bound than
%the $\mathsf{PSPACE}$ bound for the general multiplication $\bowtie$.
We do not know if the Boolean operators $(\sqcup, \sqcap, \bar{\cdot})$ in
combination with the left-sided multiplication $\lbow{\tau}$ is decidable
(existential or first order, with or without the right-sided multiplication
$\rbow{\tau}$). Determining if the existential theory with the Boolean
operators and \emph{unrestricted}
multiplication $\bowtie$ is decidable also
seems challenging.
We would also like to know if the monadic second-order
theory over these structures is decidable.
\\\\
\textbf{Acknowledgement.} We would like to thank anonymous referees for their constructive reviews. Le and Lin are partially supported by the European Research Council (ERC) under the European Union’s Horizon 2020 research and innovation programme (grant agreement no 759969). Le and Hobor are partially supported under Yale-NUS College grant R-607-265-322-121.

\bibliography{share_complex}
%\vfill
%\pagebreak
\appendix
\section{Appendix}\label{sec:formula}
\begin{figure}
\begin{verbatim}
ws2s;
pred ant(var2 Y) = 
 all1 x,y: (x~=y & x in Y & y in Y) => (~(x<=y) & ~(y<=x));
pred maxt(var2 X,var2 Y) = 
  X sub Y & ex1 r:all1 x: x in X =>
  (r <= x & all1 z: r <= z => ex1 x': x' in X & (z <= x' | x' <= z));
pred roott(var1 x,var2 X) = 
  all1 y: y in X & x <= y & all1 z:all1 y':y' in X & z <= y' => x <= z;
pred subt(var2 X, var2 Y) = 
  all1 x1:all2 X':(maxt(X',X) & roott(x1,X')) =>
  (ex2 Y':maxt(Y',Y) => roott(x1,Y'));
pred eqt(var2 X, var2 Y) = 
  subt(X,Y) & subt(Y,X);
pred singleton(var2 X) = 
  ex1 x: x in X & (all1 y: y in X => x = y);
pred uniont(var2 X,var2 Y,var2 Z) = 
  Z = X union Y & empty(X inter Y);
pred mint(var2 X) = 
  all2 Y: maxt(Y,X) => singleton(Y);
pred sub0(var2 X, var2 X0) = 
  all1 x:x in X <=> x.0 in X0;
pred sub1(var2 X, var2 X0) = 
  all1 x:x in X <=> x.1 in X0;
pred leftMul(var2 X,var2 X') = 
  all2 Y:(eqt(X,Y) & mint(Y)) => sub0(Y,X');
pred rightMul(var2 X,var2 X') = 
  all2 Y:(eqt(X,Y) & mint(Y)) => sub1(Y,X');

all2 X,X',XL,XR,XU:
  (ant(X) & ant(X') & ant(XL) & ant(XR) & ant(XU) & eqt(X,X') & 
   leftMul(X,XL) & rightMul(X,XR) & uniont(XL,XR,XU)) => (eqt(XU,X'));
\end{verbatim}
\caption{The transformation of tree share formula in \S\ref{subsec:combformpractice} into equivalent WS2S formula.}\label{fig:formula}
\end{figure}

Fig.~\ref{fig:formula} contains the MONA WS2S encoding of the following tree share formula
\[
\forall \pi, \pi'.~ \pi = \pi' ~ \Rightarrow
~ (\pi \bowtie \Tree [ $\circ$ $\bullet$ ])   \oplus (\pi \bowtie \Tree [ $\bullet$  $\circ$ ]) = \pi' .
\]
where lower case letters are for variables of binary strings and upper case
letters are for second-order monadic predicates. The last three lines in the
code are the formulas with a number of macros defined in the previous lines.
Essentially, each tree share is
represented by a second-order variable whose elements are \emph{antichains} that
describes a single path to one of its black leaves. Roughly speaking, the $\texttt{eqt}$ predicate checks whether two tree shares are equal, $\texttt{leftMul}$ and $\texttt{rightMul}$ correspond to the multiplicative predicates $\bowtie_{\Tree [ $\bullet$ $\circ$ ]}$ and $\bowtie_{\Tree [ $\circ$ $\bullet$ ]}$ respectively, and $\texttt{uniont}$ computes the additive operator $\oplus$. Other additional predicates are necessary for the consistent representation of the tree shares. In detail, $\texttt{singleton(X)}$ means that $\texttt{X}$ has exactly one element, $\texttt{ant}$ makes sure any two antichains in the same tree are neither prefix of the other, $\texttt{maxt(X,Y)}$ enforces that $\texttt{X}$ is the maximal antichain of $\texttt{Y}$, $\texttt{roott(x,X)}$ asserts $\texttt{x}$ is the root of $\texttt{X}$, $\texttt{subt}$ is a subset-like relation betweens two trees, while $\texttt{mint}$ specifies the canonical form. Lastly, we have $\texttt{sub0}$ and $\texttt{sub1}$ as the intermediate predicates for the multiplicative predicates.

\end{document}